\documentclass[11pt]{article}

\usepackage[utf8]{inputenc}
\usepackage[T1]{fontenc}
\usepackage{amsmath,amssymb,amsthm}
\usepackage{mathtools}
\usepackage[colorlinks=true,linkcolor=blue!60!black,citecolor=blue!60!black,urlcolor=blue!60!black]{hyperref}
\usepackage{cleveref}
\usepackage{enumitem}
\usepackage{booktabs}
\usepackage{listings}
\usepackage{xcolor}
\usepackage[margin=1in]{geometry}

\definecolor{leanblue}{RGB}{0,0,180}
\definecolor{leangreen}{RGB}{0,128,0}
\definecolor{leangray}{RGB}{100,100,100}
\lstdefinelanguage{Lean4}{
  keywords={theorem,def,lemma,import,open,noncomputable,section,end,
    variable,namespace,where,by,exact,apply,intro,have,let,show,
    unfold,simp,rw,calc,fun,match,with,if,then,else,structure,class,
    instance,inductive,axiom,set},
  keywordstyle=\color{leanblue}\bfseries,
  comment=[l]{--},
  morecomment=[s]{/-}{-/},
  commentstyle=\color{leangreen}\itshape,
  stringstyle=\color{leangray},
  basicstyle=\ttfamily\small,
  breaklines=true,
  showstringspaces=false,
  columns=flexible,
}
\newtheorem{theorem}{Theorem}[section]

\newtheorem{proposition}[theorem]{Proposition}
\newtheorem{corollary}[theorem]{Corollary}
\theoremstyle{definition}
\newtheorem{definition}[theorem]{Definition}
\newtheorem{remark}[theorem]{Remark}

\newcommand{\lean}[1]{\lstinline[language=Lean4,breaklines=true,breakatwhitespace=false]{#1}}
\newcommand{\RR}{\mathbb{R}}

\newcommand{\ZZ}{\mathbb{Z}}
\newcommand{\dif}{\mathrm{d}}
\newcommand{\extd}{\mathrm{d}}
\newcommand{\bdry}{\partial}
\newcommand{\Icc}[2]{[#1, #2]}
\newcommand{\mathlib}{\texttt{mathlib4}}

\title{Stokes' Theorem for Smooth Singular Cubes in Lean~4: \\
       True Pullback, Bridges to \mathlib{}, and Chain-Level $\partial^2{=}0$}

\author{David B.\ Hulak\thanks{Independent Researcher. ORCID: 0009-0002-8056-1774.}
\and Arthur F.\ Ramos\thanks{Microsoft. ORCID: 0009-0003-3568-0325.}
\and Ruy J.\ G.\ B.\ de Queiroz\thanks{Centro de Inform\'atica, Universidade Federal de Pernambuco. ORCID: 0000-0003-1482-0977.}}
\date{July 2026}

\begin{document}
\maketitle

\begin{abstract}
We present a sorry-free Lean~4/\mathlib{} formalization of Stokes' theorem for
\emph{smooth singular cubes} in arbitrary dimension:
\[
  \int_{[0,1]^{n+1}} \sigma^*(\extd\omega)
  = \sum_{i=0}^{n} (-1)^i \Bigl(
    \int_{\mathrm{face}_i(1)} \sigma^*\omega
    - \int_{\mathrm{face}_i(0)} \sigma^*\omega
  \Bigr)
\]
for any globally $C^\infty$ map $\sigma : \RR^{n+1} \to \RR^m$ and any $C^\infty$
$n$-form $\omega$ on $\RR^m$. Here $\sigma^*$ is the \emph{true} differential-form
pullback via the Fr\'echet derivative,
$(\sigma^*\omega)(x)(v_1,\ldots,v_n) = \omega(\sigma(x))(D\sigma(x)\cdot v_1,
\ldots, D\sigma(x)\cdot v_n)$. The proof reduces to box Stokes on $[0,1]^{n+1}$
via \mathlib{}'s \lean{extDeriv_pullback} (giving
$\extd(\sigma^*\omega) = \sigma^*(\extd\omega)$) and a chain-rule face-matching
identity that we formalize.

We additionally formalize: a \emph{bridge theorem} relating our coordinate
representation to \mathlib{}'s abstract \lean{extDeriv}; chain-level Stokes
extended by $\ZZ$-linearity to finite formal combinations of singular cubes;
the chain-complex identity $\partial(\partial\sigma) = 0$, proved via a
sign-reversing involution on the double-boundary index set; an underlying box
Stokes theorem for axis-aligned cubes (reduced to \mathlib{}'s divergence
theorem); dimensional specializations (FTC, rectangular Green, 3D~Gauss,
divergence consistency, integration by parts, Leibniz); and a structured
comparison with Harrison's HOL~Light formalization~\cite{harrison-stokes} of
Stokes for convex/polyhedral domains.

The development comprises 44 Lean modules (4028 source lines,
205 named declarations), is verified with zero \texttt{sorry} uses, and depends
only on the three standard Lean~4 axioms (\texttt{propext},
\texttt{Classical.choice}, \texttt{Quot.sound}), confirmed by 79 systematic
\texttt{\#print axioms} checks.
\end{abstract}


\paragraph{What this paper proves, briefly.}
For any globally $C^\infty$ map $\sigma : \RR^{n+1} \to \RR^m$ and any $C^\infty$
$n$-form $\omega$ on $\RR^m$, integrating the pullback of $\extd\omega$ over the
unit cube equals the alternating sum of pullback integrals over the $2(n{+}1)$
faces---formalized in Lean~4 with zero \lean{sorry} and only the standard
foundational axioms. The precise scope and limitations, including the relationship
to manifold Stokes, are discussed in \Cref{sec:discussion}.

\begin{table}[ht]
\centering\small
\begin{tabular}{lccc}
\toprule
\textbf{Feature} & \textbf{This paper} & \textbf{Harrison (HOL)} & \textbf{Manifold Stokes} \\
\midrule
Singular cubes             & Yes & No (polyhedral) & Via charts \\
Pullback of forms          & Yes (\lean{fderiv}) & Yes & Yes \\
Chain-level $\partial$     & Yes ($\partial^2{=}0$) & Yes (polyhedral) & Yes \\
Library integration        & Lean 4/\mathlib{} & HOL Light & N/A \\
Integration domain         & Parameter cube & Convex polytopes & Oriented manifold \\
Global smoothness required & Yes & Less restrictive & Local suffices \\
Smooth manifolds           & No  & No (convex/polyhedral) & Yes \\
Charts / atlases           & No  & No & Yes \\
Partition of unity         & No  & No & Often needed \\
\bottomrule
\end{tabular}
\caption{Scope comparison: this work is a cubical-chain formalization with true pullback,
not a manifold-level Stokes theorem. Rows are ordered by coverage in this paper.}
\label{tab:scope}
\end{table}

\section{Introduction}
\label{sec:intro}

The generalized Stokes' theorem~\cite{spivak,lee-manifolds} is one of the central results in
differential geometry,
unifying the fundamental theorem of calculus, Green's theorem, the classical divergence
theorem, and the curl theorem into a single statement: for an oriented smooth manifold~$M$
with boundary and a differential form~$\omega$ of appropriate degree,
\begin{equation}\label{eq:stokes}
  \int_M \extd\omega = \int_{\bdry M} \omega
\end{equation}
Prior to this work, to our knowledge no complete formalization of the
differential-form Stokes theorem (for smooth manifolds with boundary, singular chains,
or comparable generality) exists in Lean~4~\cite{lean4}/\mathlib{}~\cite{mathlib,mathlib-community}
as of April~2026.
Harrison~\cite{harrison-stokes} formalized Stokes' theorem in HOL Light for
convex and polyhedral domains, building on the Euclidean space
foundations of~\cite{Harrison2005}. Harrison's published development does not
cover smooth manifolds with boundary
(see \Cref{sec:harrison} for a detailed comparison).

In this paper, we present a sorry-free Lean~4 formalization of Stokes' theorem
at two levels:
\begin{enumerate}
  \item \textbf{Smooth singular cubes}: For any $C^\infty$ map
    $\sigma : \RR^{n+1} \to \RR^m$ and smooth $n$-form $\omega$ on $\RR^m$,
    the integral of the pullback $\sigma^*(\extd\omega)$ over the unit cube
    equals the oriented sum of face integrals of $\sigma^*\omega$.
    The pullback is the \emph{true} differential-form pullback, using the
    Fr\'echet derivative of $\sigma$.
  \item \textbf{Axis-aligned boxes}: For any $C^\infty$ coordinate $n$-form
    on~$\RR^{n+1}$ and any box $\Icc{a}{b}$, the integral of $\extd\omega$
    (evaluated on the standard frame) equals the boundary integral.
\end{enumerate}

The singular cube theorem (\lean{singularStokes}) allows smooth maps
$\sigma : \RR^{n+1} \to \RR^m$ that can parametrize curves, surfaces, and
higher-dimensional objects (possibly degenerate or self-overlapping), with
integration always over the parameter cube---broader as a
parametrized-chain statement than axis-aligned boxes alone, though still
short of smooth manifolds with boundary or image-domain theory.

\subsection{Contributions}

Our central contribution is the singular-cube theorem; the box, bridge, and
chain results form the technical infrastructure that makes the singular case
go through.

\paragraph{Headline contributions.}
\begin{enumerate}[label=(\roman*),leftmargin=*]
  \item \textbf{Stokes for smooth singular cubes} with true pullback via
    \lean{fderiv} (\Cref{sec:singular-stokes}). The proof reduces to box
    Stokes via \mathlib{}'s \lean{extDeriv_pullback}, which gives
    $\extd(\sigma^*\omega) = \sigma^*(\extd\omega)$, and a chain-rule
    \emph{face-matching identity} (\Cref{thm:face-matching}) that we formalize.
  \item \textbf{Chain-level $\partial^2{=}0$} for singular cubical chains, proved
    via a sign-reversing fixed-point-free involution on the double-boundary
    index set (\Cref{thm:bdry2-chain}); chain-level Stokes is also formalized
    by extending the singular theorem $\ZZ$-linearly over \lean{Finsupp}.
  \item A \textbf{formal bridge} from our coordinate representation to
    \mathlib{}'s abstract \lean{extDeriv} (\Cref{sec:bridge}), proving that our
    coordinate formula is the evaluation of the abstract exterior derivative
    on standard basis vectors, and a unified statement (\lean{stokes_extDeriv_smooth})
    that needs only one \lean{ContDiff} hypothesis.
\end{enumerate}

\paragraph{Supporting infrastructure.}
\begin{enumerate}[label=(\roman*),resume,leftmargin=*]
  \item A \textbf{sorry-free} proof of box Stokes for axis-aligned cubes in
    arbitrary dimension $n+1$ (\Cref{sec:main-theorem}), reducing to
    \mathlib{}'s divergence theorem~\cite{kudryashov-divergence} via a
    signed-coefficient construction (short proof body; the work is in the definitions).
  \item \textbf{Box chain infrastructure}: $\ZZ$-linear combinations of boxes
    with Stokes extended by linearity, two-box additivity, and shared-face
    cancellation (\Cref{sec:chains}).
  \item \textbf{Box-level $\partial^2{=}0$ sign and face cancellation}: for
    each codimension-2 face, the two boundary paths produce the same face box
    with opposite signs---proved without defining $\partial$ as a chain
    homomorphism on undecorated boxes.
  \item \textbf{Dimensional specializations}: FTC ($n{=}0$), rectangular
    Green's theorem ($n{=}1$, with four-edge decomposition), divergence
    consistency, and 3D Gauss; plus calculus corollaries (integration by
    parts, Leibniz product rule). All are direct instantiations of
    \lean{stokes_smooth}, not independent derivations.
  \item A \textbf{1D agreement check} (\Cref{sec:validation}) showing that
    the Stokes-derived FTC equals \mathlib{}'s interval-integral FTC.
  \item A \textbf{structured comparison} with Harrison's HOL Light Stokes for
    convex/polyhedral domains (\Cref{sec:harrison}), and an honest assessment
    of what remains for full manifold Stokes (\Cref{sec:discussion}).
\end{enumerate}

\paragraph{Theorem dependency transparency.}
\lean{stokes_on_box} is a direct wrapper around \mathlib{}'s divergence theorem
(the mathematical content is the signed-coefficient construction).
\lean{stokes_smooth} derives hypotheses from $C^\infty$ and calls \lean{stokes_on_box}.
\lean{singularStokes} reduces to \lean{stokes_smooth} via \lean{extDeriv_pullback}
and the face-matching identity.
\lean{stokes_singular_boundary} restates \lean{singularStokes} with the boundary
integral on the left; \lean{stokes_singular_chain} is the corresponding
formal-chain version using \lean{singularBoundarySingle}.
\lean{stokes_chain} extends by $\ZZ$-linearity
(\lean{Finsupp.induction} for singular chains, \lean{Finsupp.sum_congr} for box chains).
The FTC, Green, divergence, and Gauss results are dimensional instantiations of
\lean{stokes_smooth}.
\lean{dd_zero_abstract} is a direct invocation of \mathlib{}'s \lean{extDeriv_extDeriv}.
\lean{bdry_sq_sign_cancel} is independent sign arithmetic (no analysis dependency;
uses only \lean{omega} and ring lemmas).

\subsection{Relation to Existing Work}

\paragraph{Harrison's HOL Light formalization.}
Harrison~\cite{harrison-stokes} formalized a version of Stokes' theorem in HOL Light
for \emph{convex and polyhedral domains} in $\RR^N$, extended to polyhedral chains.
Contrary to what some summaries suggest, Harrison's formalization does \emph{not}
cover arbitrary smooth manifolds with boundary; it operates on convex polytopes with
the Henstock--Kurzweil gauge integral.
With our singular cube extension, the mathematical scope of the two developments
is \emph{complementary but different}: Harrison proves Stokes for convex/polyhedral domains
in $\RR^n$ extended to polyhedral chains,
while our work uses smooth singular cubical parametrizations from the unit cube.
The key infrastructure differences remain:
(1)~Harrison works in HOL's simple type theory; we use Lean~4's dependent types;
(2)~Harrison builds all integration from scratch; we reduce to \mathlib{}'s
divergence theorem;
(3)~we provide a bridge to \mathlib{}'s abstract \lean{extDeriv} and use
\lean{extDeriv_pullback} for the singular theorem; we are not aware of a
directly comparable bridge to an independently developed abstract
exterior-derivative API in HOL~Light, since HOL~Light does not maintain such
an independent abstract exterior algebra alongside its integration theory.
A detailed structural comparison appears in \Cref{sec:harrison}.

\paragraph{Mathlib infrastructure.}
As of the \mathlib{} version pinned in this project (April 2026), the library provides:
\begin{itemize}
  \item The box divergence theorem (\lean{integral_divergence_of_hasFDerivAt_off_countable'}),
    which uses the Bochner set integral despite its historical connection with
    Henstock--Kurzweil-style divergence theorems~\cite{kudryashov-divergence}.
    Our box Stokes is a direct wrapper around this result.
  \item The abstract exterior derivative (\lean{extDeriv}) for differential forms~\cite{kudryashov-extDeriv}
  \item Naturality of pullback (\lean{extDeriv_pullback})~\cite{kudryashov-extDeriv}
  \item Continuous alternating maps (\lean{ContinuousAlternatingMap})
  \item Smooth manifold infrastructure (\lean{SmoothManifoldWithCorners})~\cite{gouezel-manifolds}
\end{itemize}
However, \mathlib{} does \emph{not} yet provide:
integration of differential forms on manifolds,
orientation theory for manifolds with boundary,
or a differential-form Stokes theorem comparable to the present statement
(though it does provide the box divergence theorem
from which our box Stokes reduces).
Our work provides the first Lean~4 formalization of Stokes' theorem
for smooth singular cubes, contributing to the
growing body of Lean~4 formalizations reported in venues such as
AFM~\cite{best-flt} and ITP~\cite{van-doorn-itp}.

\paragraph{Broader ITP landscape.}
Beyond Harrison's HOL Light work, integration and measure theory have been formalized
extensively in Isabelle/HOL by H\"olzl, Immler, and others~\cite{hoelzl-measure},
providing foundational infrastructure for Lebesgue integration. The closest
non-HOL-Light formalization of a Stokes-type theorem is Abdulaziz and
Paulson's Isabelle/HOL formalization of \emph{Green's theorem on
2-dimensional cells}~\cite{abdulaziz-green}, which is a 2D specialization of
Stokes; we extend that line of work to arbitrary dimension and to true
differential-form pullback under smooth singular parametrizations.
We are not aware of a Coq, Mizar, Metamath, or HOL4 formalization of Stokes'
theorem for singular cubes or smooth manifolds with boundary.

\paragraph{Paper organization.}
\Cref{sec:singular-stokes} presents the main singular Stokes theorem first for
motivation; it depends on the box infrastructure of \S\S\ref{sec:math}--\ref{sec:bridge},
which is developed subsequently.
\S\S\ref{sec:classical}--\ref{sec:dd-zero} give specializations and extensions.
\S\S\ref{sec:architecture}--\ref{sec:discussion} discuss architecture, insights, and limitations.

\section{Smooth Singular Cubical Stokes}
\label{sec:singular-stokes}

The main result of this paper is Stokes' theorem for smooth singular
cubes---a generalization of the box Stokes theorem that allows the integrand
to come from a globally smooth map $\sigma:\RR^{n+1}\to\RR^m$, with the
integral still taken over the parameter cube $[0,1]^{n+1}$ (the box case is
recovered as $\sigma=\mathrm{id}$ on $\RR^{n+1}$). No image-domain theory is
formalized.

\begin{remark}[On the term ``smooth singular cube'']\label{rem:terminology}
In the algebraic-topology literature, a singular $n$-cube is usually a continuous
map from the standard cube $I^n$ into a topological space; smoothness is
imposed only on $I^n$. \emph{Our formal objects are globally smooth extensions}:
the map $\sigma$ is required to be $C^\infty$ on \emph{all of $\RR^{n+1}$},
and integration is over $[0,1]^{n+1}$. This is mathematically stronger than
the conventional definition---any map smooth on the closed cube with compatible
jet data admits a smooth extension by Whitney's theorem~\cite{whitney}, but we impose the
extension as a hypothesis rather than constructing it. This choice substantially
simplifies the formalization (\Cref{sec:insights}); we discuss the trade-off
explicitly in \Cref{sec:discussion}.
\end{remark}

\subsection{Smooth Singular Cubes}

\begin{definition}[Smooth singular cube]\label{def:singular-cube}
  A \emph{smooth singular $(n{+}1)$-cube in $\RR^m$} is a globally $C^\infty$ map
  $\sigma : \RR^{n+1} \to \RR^m$, integrated over the unit cube $[0,1]^{n+1}$.
\end{definition}

In Lean~4:\footnote{Lean code excerpts use ASCII equivalents of Unicode operators
(e.g., \texttt{->} for $\to$, \texttt{forall} for $\forall$) for typesetting
compatibility. The repository uses standard Unicode Lean~4 syntax.}
\begin{lstlisting}
structure SmoothSingularCube (n m : Nat) where
  toFun : (Fin n -> Real) -> (Fin m -> Real)
  smooth : ContDiff Real Top toFun
\end{lstlisting}

\begin{remark}
In the Lean code, \lean{SmoothSingularCube n m} takes the cube dimension directly:
an $(n{+}1)$-cube has type \lean{SmoothSingularCube (n+1) m} with domain $\mathrm{Fin}(n{+}1) \to \RR$.
The smoothness hypothesis is \emph{global} ($C^\infty$ on all of $\RR^{n+1}$),
not just on the unit cube. This is a stronger hypothesis than strictly necessary
but simplifies the formalization considerably, as \lean{ContDiff} at the top level
derives all needed differentiability properties automatically.
\end{remark}

\begin{definition}[Face inclusion and singular faces]
  The \emph{face inclusion} $\iota_{i,\varepsilon} : \RR^n \to \RR^{n+1}$ inserts the
  constant $\varepsilon$ at position~$i$:
  \[
    \iota_{i,\varepsilon}(t)_j =
    \begin{cases}
      t_j & \text{if } j < i \\
      \varepsilon & \text{if } j = i \\
      t_{j-1} & \text{if } j > i
    \end{cases}
  \]
  The \emph{$i$-th face at $\varepsilon$} of a singular cube $\sigma$ is the composition
  $\sigma \circ \iota_{i,\varepsilon}$, itself a smooth singular $n$-cube.
\end{definition}

\begin{lstlisting}
def faceInclusion (i : Fin (n + 1)) (eps : Real) (t : Fin n -> Real) :
    Fin (n + 1) -> Real := ...

def singularFace (sigma : SmoothSingularCube (n + 1) m)
    (i : Fin (n + 1)) (eps : Real) : SmoothSingularCube n m where
  toFun := sigma.toFun comp faceInclusion i eps
  smooth := sigma.smooth.comp (faceInclusion_contDiff i eps)
\end{lstlisting}

\subsection{True Pullback via Fr\'echet Derivative}

The central construction of the singular cube formalization is the use of the
\emph{true} differential-form pullback, defined via the Fr\'echet derivative.
We write $d$ for the domain dimension.
In the Stokes application, $d = n{+}1$ and $k = n$, but the pullback definition is general:

\begin{definition}[Pullback form]\label{def:pullback}
  For a smooth singular cube $\sigma : \RR^d \to \RR^m$ and a $k$-form
  $\omega$ on $\RR^m$, the pullback $\sigma^*\omega$ is the $k$-form on $\RR^d$ defined by:
  \[
    (\sigma^*\omega)(x)(v_1, \ldots, v_k) =
      \omega(\sigma(x))(D\sigma(x) \cdot v_1, \ldots, D\sigma(x) \cdot v_k)
  \]
  where $D\sigma(x) = \mathrm{fderiv}\;\RR\;\sigma\;x$ is the Fr\'echet derivative.
\end{definition}

In Lean~4, this is expressed via \lean{ContinuousAlternatingMap.compContinuousLinearMap}:%
\footnote{The type \texttt{E [Alt\^{}Fin k]-\textgreater{}L[R] F} denotes
\lean{ContinuousAlternatingMap R E F (Fin k)}---a continuous $k$-linear alternating map.}
\begin{lstlisting}
def pullbackForm (sigma : SmoothSingularCube d m)
    (omega : (Fin m -> Real) ->
         (Fin m -> Real) [Alt^Fin k]->L[Real] Real) :
    (Fin d -> Real) -> (Fin d -> Real) [Alt^Fin k]->L[Real] Real :=
  fun x => (omega (sigma.toFun x)).compContinuousLinearMap
              (fderiv Real sigma.toFun x)
\end{lstlisting}

This is genuinely the differential-geometric pullback: it applies the form $\omega$
at the image point $\sigma(x)$ to vectors pushed forward by the derivative $D\sigma(x)$.
This contrasts with the ``coefficient precomposition'' $(A \cdot \omega)_i(x) := \omega_i(Ax)$
used in the box infrastructure, which does not mix coefficients or involve determinants.

\begin{definition}[Integration of forms over singular cubes]
  The integral of a $d$-form $\omega$ over a singular $d$-cube $\sigma$ is:
  \[
    \int_\sigma \omega := \int_{[0,1]^d}
      (\sigma^*\omega)(x)(e_0, \ldots, e_{d-1})\, \dif x
  \]
  where $e_j = \mathrm{Pi.single}\;j\;1$ are the standard basis vectors (zero-indexed).
\end{definition}

\begin{lstlisting}
def integrateForm (sigma : SmoothSingularCube d m)
    (omega : (Fin m -> Real) ->
         (Fin m -> Real) [Alt^Fin d]->L[Real] Real) : Real :=
  integral (x in Icc (fun _ => 0) (fun _ => 1))
    (pullbackForm sigma omega x) (fun j => Pi.single j 1)
\end{lstlisting}

\subsection{Proof Strategy}

The proof of singular Stokes proceeds in three steps:

\paragraph{Step 1: Naturality of pullback.}
The key identity is \emph{commutativity of pullback with the exterior derivative}:
\[
  \extd(\sigma^*\omega) = \sigma^*(\extd\omega)
\]
In Lean~4, this follows directly from \mathlib{}'s \lean{extDeriv_pullback}:

\begin{lstlisting}
theorem pullback_extDeriv (sigma : SmoothSingularCube (n + 1) m)
    (omega : ...) (h_omega : ContDiff Real Top omega)
    (x : Fin (n + 1) -> Real) :
    extDeriv (pullbackForm sigma omega) x =
    (extDeriv omega (sigma.toFun x)).compContinuousLinearMap
      (fderiv Real sigma.toFun x)
\end{lstlisting}

\paragraph{Step 2: Apply box Stokes to $\sigma^*\omega$.}
Since $\sigma^*\omega$ is a smooth differential form on $\RR^{n+1}$, the existing
box Stokes theorem (\lean{stokes_extDeriv}, \Cref{sec:main-theorem}) applies on $[0,1]^{n+1}$:
\[
  \int_{[0,1]^{n+1}} \extd(\sigma^*\omega)
  \;=\; \int_{\partial[0,1]^{n+1}} \sigma^*\omega
\]
(in Lean: \lean{boxIntegral (extDerivCoord ...) = bdryIntegral ...}).
Smoothness of the coordinate coefficients of $\sigma^*\omega$ is proved via
\lean{toCoordNForm_pullback_isSmooth}, which derives it from the global
$C^\infty$ hypotheses on $\sigma$ and $\omega$ using the chain rule for
\lean{ContDiff}.

\paragraph{Step 3: Face matching.}
The face matching identity connects the box boundary integral of $\sigma^*\omega$
to the face integrals of the singular cube. The core lemma is:

\begin{theorem}[Face matching]\label{thm:face-matching}
  For $\sigma$, $\omega$, face index $i$, value $\varepsilon$, and point $x \in \RR^n$,
  writing $\mathrm{succAbove}(i,k)$ for the $k$-th element of $\{0,\ldots,n\} \setminus \{i\}$
  in increasing order (the order-preserving injection
  $\{0,\ldots,n{-}1\} \hookrightarrow \{0,\ldots,n\} \setminus \{i\}$):
  \[
    (\sigma^*\omega)(\iota_{i,\varepsilon}(x))(e_{\mathrm{succAbove}(i,0)}, \ldots,
      e_{\mathrm{succAbove}(i,n-1)})
    = ((\sigma \circ \iota_{i,\varepsilon})^*\omega)(x)(e_0, \ldots, e_{n-1})
  \]
\end{theorem}

The proof uses the chain rule: $D(\sigma \circ \iota_{i,\varepsilon})(x)
= D\sigma(\iota_{i,\varepsilon}(x)) \circ D\iota_{i,\varepsilon}(x)$,
together with the derivative identity
$D\iota_{i,\varepsilon}(x) \cdot e_j = e_{\mathrm{succAbove}(i,j)}$
(proved as \lean{fderiv_faceInclusion_single}).
The alternating map's multilinearity then transforms the $n$ basis vectors
$e_{\mathrm{succAbove}(i,0)}, \ldots, e_{\mathrm{succAbove}(i,n{-}1)}$
applied inside $\sigma^*\omega$ into $e_0, \ldots, e_{n-1}$ applied inside
$(\sigma \circ \iota_{i,\varepsilon})^*\omega$.

\subsection{The Main Theorem}

\begin{theorem}[Singular cubical Stokes, \lean{singularStokes}]\label{thm:singular-stokes}
  Let $\sigma : \RR^{n+1} \to \RR^m$ be a smooth singular $(n{+}1)$-cube and
  $\omega$ a $C^\infty$ $n$-form on $\RR^m$. Then:
  \[
    \int_\sigma \extd\omega
    = \sum_{i=0}^{n} (-1)^i \Bigl(
      \int_{\mathrm{face}_i(1)} \omega
      - \int_{\mathrm{face}_i(0)} \omega
    \Bigr)
  \]
  where all integrals use the true pullback and the unit cube domain.
\end{theorem}

\begin{proof}[Proof sketch]
\begin{sloppypar}
The argument proceeds in three steps.
(1)~\emph{Naturality}: by \mathlib{}'s
\lean{extDeriv_pullback}, applied to the
globally $C^\infty$ hypotheses on $\sigma$ and $\omega$, we have
$\extd(\sigma^*\omega)(x) = (\sigma^*(\extd\omega))(x)$ pointwise.
The pullback uses \lean{fderiv}, and the hypotheses
are satisfied because $\sigma$ is globally $C^\infty$
and $\omega$ is $C^\infty$ as a
\texttt{ContinuousAlternatingMap}-valued function.
(2)~\emph{Box Stokes}: integrating both sides over $[0,1]^{n+1}$ and applying the
existing box Stokes theorem to $\sigma^*\omega$ (whose coordinate coefficients
are smooth by the chain rule), the left-hand side equals the boundary integral
$\sum_i (-1)^i (F_i(1) - F_i(0))$ where $F_i(\varepsilon)$ is the box-boundary
integral on the $i$-th face at value~$\varepsilon$.
(3)~\emph{Face matching}: the face-matching identity (\Cref{thm:face-matching})
converts each box-boundary term into the corresponding singular-face integral,
matching the right-hand side.
\end{sloppypar}
\end{proof}

\medskip\noindent
The Lean statement:
\begin{lstlisting}
theorem singularStokes (sigma : SmoothSingularCube (n + 1) m)
    (omega : (Fin m -> Real) ->
         (Fin m -> Real) [Alt^Fin n]->L[Real] Real)
    (h_omega : ContDiff Real Top omega) :
    integrateForm sigma (fun y => extDeriv omega y) =
    sum i : Fin (n + 1),
      (-1 : Real) ^ (i : Nat) *
      (integrateForm (singularFace sigma i 1) omega -
       integrateForm (singularFace sigma i 0) omega)
\end{lstlisting}

\paragraph{The theorem applies non-trivially to non-axis-aligned parametrizations.}
The theorem holds for any globally smooth $\sigma$, not only for box-aligned
parametrizations. As an informal illustration (this specific computation is
not formalized; it serves as a paper-only sanity check), take the polar
parametrization of an annulus,
\[
  \sigma : \RR^2 \to \RR^2,\qquad
  \sigma(r,\theta) =
  \bigl((1+r)\cos(2\pi\theta),\; (1+r)\sin(2\pi\theta)\bigr),
\]
together with the 1-form
$\omega = \tfrac{1}{2}(x\,\dif y - y\,\dif x)$ on $\RR^2$ (so
$\extd\omega = \dif x\wedge\dif y$). Because $\sigma$ is real-analytic on all
of $\RR^2$, the global $C^\infty$ hypothesis of \Cref{thm:singular-stokes}
is satisfied. The theorem then equates the parameter-cube integral
$\int_{[0,1]^2}\sigma^*(\extd\omega)$ with the alternating sum of pullback
integrals over the four edges of the unit square: by direct computation, the
parameter-cube integrand reduces to the Jacobian
$\det D\sigma(r,\theta) = 2\pi(1+r)$, consistent with the classical
Jacobian formula for area.
Concretely, $\int_0^1\!\int_0^1 2\pi(1+r)\,dr\,d\theta = 3\pi$, and
the alternating face-integral sum yields the same value (the two
$\theta{=}0$ and $\theta{=}1$ edge integrals cancel by periodicity)---providing
a numerical sanity check of the theorem.
The development includes
\texttt{LeanStokes/CubeStokes/Demo.lean}, which exercises box Stokes (the
specialization of \Cref{thm:singular-stokes} to $\sigma = \mathrm{id}$) in
4D, 5D, and 10D; the full singular theorem with non-trivial~$\sigma$ is
exercised at the type-checking level via the \lean{singularStokes} declaration.

\subsection{Chain-Level Extension}

Stokes extends by $\ZZ$-linearity to formal chains of singular cubes:

\begin{definition}[Singular chain]
  A \emph{smooth singular $n$-chain} in $\RR^m$ is a formal $\ZZ$-linear combination
  of smooth singular $n$-cubes, implemented as a finitely-supported function
  (\lean{Finsupp}) from singular cubes to $\ZZ$:
  \[
    C_n(\RR^m) := \left\{ \sum_{\text{finite}} c_\sigma \cdot \sigma \;\middle|\;
      c_\sigma \in \ZZ,\; \sigma \text{ a smooth singular $n$-cube} \right\}
  \]
\end{definition}

\begin{theorem}[Stokes for singular boundary, \lean{stokes_singular_boundary}]\label{thm:stokes-chain-singular}
\begin{sloppypar}
  For a smooth singular cube $\sigma$ and smooth form $\omega$:
  \[
    \int_{\bdry\sigma} \omega = \int_\sigma \extd\omega
  \]
  where $\int_{\bdry\sigma}\omega = \sum_i (-1)^i (\int_{\mathrm{face}_i(1)}\omega
  - \int_{\mathrm{face}_i(0)}\omega)$.
  This is a direct restatement of \lean{singularStokes} (swapped sides)
  using \lean{bdryIntegral_singular}. The corresponding formal-chain version,
  \lean{stokes_singular_chain}, equates
  \lean{integrateChain (singularBoundarySingle sigma) omega}
  with the interior integral.
\end{sloppypar}
\end{theorem}

\begin{remark}
We have a boundary operator on individual cubes (\lean{singularBoundarySingle})
and chain-level Stokes (\lean{stokes_chain}) extending by linearity to
formal chains. The chain complex identity $\partial^2 = 0$ is proved as
\lean{bdry_bdry_chain_zero} (\Cref{sec:bdry-sq-singular}).
\end{remark}

\begin{sloppypar}
Chain integration satisfies additivity
(\lean{integrateChain_add}) and scalar compatibility
(\lean{integrateChain_smul}):
\end{sloppypar}
\[
  \int_{c_1 + c_2} \omega = \int_{c_1}\omega + \int_{c_2}\omega
  \qquad\text{and}\qquad
  \int_{k \cdot c}\omega = k \cdot \int_c \omega
\]
where $k \in \ZZ$ acts via the canonical embedding $\ZZ \hookrightarrow \RR$.

\section{Box Stokes: Mathematical Framework}
\label{sec:math}

The singular Stokes theorem reduces to the box Stokes theorem, which we now describe.
This section presents the coordinate-level infrastructure that serves as the
foundation for the entire development.

\subsection{Coordinate \texorpdfstring{$n$}{n}-Forms on \texorpdfstring{$\RR^{n+1}$}{R\^{}(n+1)}}

An $n$-form on $\RR^{n+1}$ can be written in coordinates as:
\[
  \omega = \sum_{i=0}^{n} \omega_i(x)\, \dif x_0 \wedge \cdots \wedge \widehat{\dif x_i}
    \wedge \cdots \wedge \dif x_n
\]
where $\widehat{\dif x_i}$ denotes omission. We represent this as a family of
coefficient functions:

\begin{definition}[Coordinate $n$-form]
  A \emph{coordinate $n$-form} on $\RR^{n+1}$ is a function
  $\omega : \mathrm{Fin}(n{+}1) \to (\RR^{n+1} \to \RR)$
  where $\omega_i(x)$ is the coefficient of the $i$-th basis $n$-form at~$x$.
\end{definition}

In Lean~4:
\begin{lstlisting}
def CoordNForm (n : Nat) := Fin (n + 1)  ->  (Fin (n + 1)  ->  Real)  ->  Real
\end{lstlisting}

\subsection{Exterior Derivative}

The exterior derivative of $\omega$ is the $(n{+}1)$-form (i.e., a top form on $\RR^{n+1}$):
\[
  \extd\omega = \left(\sum_{i=0}^{n} (-1)^i \frac{\bdry \omega_i}{\bdry x_i}\right)
    \dif x_0 \wedge \cdots \wedge \dif x_n
\]

The key insight for our proof is the \emph{signed coefficient} construction:
\begin{definition}[Signed coefficient]
  $f_i(x) := (-1)^i \cdot \omega_i(x)$
\end{definition}

\noindent The sign arises because $\extd(\omega_i\,\dif x_0\wedge\cdots
\wedge\widehat{\dif x_i}\wedge\cdots\wedge\dif x_n)$ requires $i$ adjacent
transpositions to place $\dif x_i$ in position~$i$, yielding the factor $(-1)^i$.
With this convention, the exterior derivative coefficient becomes the
\emph{divergence} of the vector field $(f_0, \ldots, f_n)$:
\[
  \extd\omega = \left(\sum_{i=0}^{n} \frac{\bdry f_i}{\bdry x_i}\right)
    \dif x_0 \wedge \cdots \wedge \dif x_n
\]

This is precisely what \mathlib{}'s divergence theorem computes.

\subsection{Boundary Integral}

For a box $\Icc{a}{b} = \prod_{i=0}^{n} [a_i, b_i] \subset \RR^{n+1}$,
the oriented boundary integral of $\omega$ is:
\[
  \int_{\bdry\Icc{a}{b}} \omega = \sum_{i=0}^{n} \left(
    \int_{\text{face}_i} f_i(b_i, \hat{x}_i)\, \dif\hat{x}_i
    - \int_{\text{face}_i} f_i(a_i, \hat{x}_i)\, \dif\hat{x}_i
  \right)
\]
where $\text{face}_i = \prod_{j \neq i} [a_j, b_j]$ is the $n$-dimensional
face box, $\hat{x}_i = (x_0, \ldots, x_{i-1}, x_{i+1}, \ldots, x_n)$ denotes
all coordinates except~$i$, and $f_i(c, \hat{x}_i)$ denotes inserting $c$ at position~$i$.

This exactly matches the right-hand side of \mathlib{}'s divergence theorem.

\section{Box Stokes: Main Theorem}
\label{sec:main-theorem}

\begin{theorem}[Cubical Stokes on boxes, full-hypothesis version]\label{thm:stokes-box}
  Let $\omega$ be a coordinate $n$-form on $\RR^{n+1}$, $a \leq b$ component-wise,
  and $s$ a countable subset. If the signed coefficients $f_i = (-1)^i \omega_i$ are:
  \begin{enumerate}
    \item continuous on $\Icc{a}{b}$,
    \item differentiable on the interior minus~$s$, and
    \item the exterior derivative is integrable on~$\Icc{a}{b}$,
  \end{enumerate}
  then $\displaystyle\int_{\Icc{a}{b}} \extd\omega = \int_{\bdry\Icc{a}{b}} \omega$.
\end{theorem}

The Lean statement:
\begin{lstlisting}
theorem stokes_on_box (a b : Fin (n + 1)  ->  Real) (hle : a  <=  b)
    (omega : CoordNForm n)
    (s : Set (Fin (n + 1)  ->  Real)) (hs : s.Countable)
    (hc : forall  i, ContinuousOn (signedCoeff omega i) (Icc a b))
    (hd : forall  x  in  (pi univ fun i => Ioo (a i) (b i)) \ s,
      forall  i, HasFDerivAt (signedCoeff omega i)
        ((-1 : Real) ^ (i : Nat)  smul  fderiv Real (omega i) x) x)
    (hi : IntegrableOn (fun x => Finset.sum i : Fin (n + 1),
      ((-1 : Real) ^ (i : Nat)  smul  fderiv Real (omega i) x) (Pi.single i 1))
      (Icc a b)) :
    boxIntegral (extDerivCoord omega) a b = bdryIntegral omega a b
\end{lstlisting}

\begin{proof}[Proof (3 lines)]
The proof unfolds definitions and applies the divergence theorem directly:
\begin{lstlisting}
  unfold boxIntegral bdryIntegral extDerivCoord
  exact integral_divergence_of_hasFDerivAt_off_countable'
    a b hle (signedCoeff omega)
    (fun i x => (-1 : Real) ^ (i : Nat)  smul  fderiv Real (omega i) x)
    s hs hc hd hi
\end{lstlisting}
The key mathematical content is in the \emph{definitions}: choosing
\lean{signedCoeff} so that the divergence theorem's hypotheses align exactly
with our boundary integral.
\end{proof}

\begin{theorem}[Smooth version]\label{thm:stokes-smooth}
  If all coefficient functions $\omega_i$ are $C^\infty$ (i.e., \lean{ContDiff Real Top (omega i)}
  for all~$i$, abbreviated \lean{IsSmooth omega} in the formalization),
  then Stokes holds with no countable exceptional set:
\begin{lstlisting}
theorem stokes_smooth (a b : Fin (n + 1)  ->  Real) (hle : a  <=  b)
    (omega : CoordNForm n) (homega : IsSmooth omega) :
    boxIntegral (extDerivCoord omega) a b = bdryIntegral omega a b
\end{lstlisting}
  The proof derives the hypotheses of \Cref{thm:stokes-box} (integrability,
  differentiability off a countable set) from \lean{ContDiff} using standard
  \mathlib{} lemmas, then applies the box theorem.
\end{theorem}

\section{Bridge: From \mathlib{} Exterior Derivative to Box Stokes}
\label{sec:bridge}

The bridge theorem validates our coordinate formula against mathlib's
independently-defined abstract exterior derivative. It proves that evaluating
\mathlib{}'s \lean{extDeriv}---defined via Fr\'echet derivatives of alternating
maps---on standard basis vectors recovers our coordinate
\lean{extDerivCoord}, confirming that the sign convention is mathematically
correct. This does not establish a full theory of form integration or pullbacks,
but it provides the independent consistency guarantee a referee would expect.

\begin{definition}[Coefficient extraction]
  Given a \mathlib{} differential form $\omega : \RR^{n+1} \to \Lambda^n(\RR^{n+1}, \RR)$,
  its $i$-th coordinate coefficient is:
  \[
    (\mathrm{toCoordNForm}\;\omega)_i(x) = \omega(x)(e_{s(i,0)}, \ldots, e_{s(i,n-1)})
  \]
  where $s(i,k) = \mathrm{succAbove}(i, k)$ enumerates $\{0,\ldots,n\}\setminus\{i\}$.
\end{definition}

\begin{theorem}[Bridge theorem]\label{thm:bridge}
  For a form $\omega$ differentiable at $x$:
  \[
    (\extd\omega)(x)(e_0, \ldots, e_n)
      = \mathrm{extDerivCoord}(\mathrm{toCoordNForm}\;\omega)(x)
  \]
\end{theorem}

The proof uses \lean{extDeriv_apply_vectorField_of_pairwise_commute} with constant
vector fields $V_j(x) = e_j$, which have zero Lie bracket.
This \mathlib{} lemma expresses $(\extd\omega)(V_0, \ldots, V_n)$ as an alternating sum
of directional derivatives $V_j[\omega(V_0, \ldots, \widehat{V_j}, \ldots, V_n)]$;
the Lie bracket terms vanish because the $V_j$ are constant:
\begin{lstlisting}
theorem extDeriv_topCoeff_eq_extDerivCoord
    (omega : (Fin (n + 1)  ->  Real)  -> 
         (Fin (n + 1)  ->  Real) [Alt^Fin n] -> L[Real] Real)
    (x : Fin (n + 1)  ->  Real)
    (homega : DifferentiableAt Real omega x) :
    extDeriv omega x (fun j => Pi.single j 1)
      = extDerivCoord (toCoordNForm omega) x
\end{lstlisting}

\begin{theorem}[Unified Stokes on boxes via \lean{extDeriv}]\label{thm:unified}
  Combining the bridge with cubical Stokes:
  \[
    \int_{\Icc{a}{b}} (\extd\omega)(e_0,\ldots,e_n) = \int_{\bdry\Icc{a}{b}} \omega
  \]
  for any $C^\infty$ differential form $\omega$. The cleanest Lean statement is:
\begin{lstlisting}
theorem stokes_extDeriv_smooth
    (omega : (Fin (n + 1)  ->  Real)  -> 
         (Fin (n + 1)  ->  Real) [Alt^Fin n] -> L[Real] Real)
    (a b : Fin (n + 1)  ->  Real) (hle : a  <=  b)
    (homega : ContDiff Real Top omega) :
    integral (x in Icc a b) (extDeriv omega x (fun j => Pi.single j 1))
      = bdryIntegral (toCoordNForm omega) a b
\end{lstlisting}
  This requires only a single hypothesis: global $C^\infty$ smoothness of $\omega$.
  Differentiability and coefficient smoothness are derived automatically.
\end{theorem}

\section{Dimensional Specializations}
\label{sec:classical}

Each result below is a box/interval specialization of \lean{stokes_smooth}
(a direct dimensional instantiation, not an independent derivation):

\begin{corollary}[Fundamental Theorem of Calculus, \lean{ftc_stokes}]\label{cor:ftc}
  For smooth $f : \RR \to \RR$ and $a \leq b$:
  \[
    \int_{[a,b]} f'(x)\,\dif x = f(b) - f(a)
  \]
  Proved by defining $\omega_0(y) = f(y_0)$ for $y \in \RR^1$ and applying
  \lean{stokes_smooth} at $n=0$. The boundary integral reduces via
  \lean{bdryIntegral_fin1} (which uses \lean{Fin.insertNth_apply_same} to
  evaluate face values) and the volume integral via \lean{extDerivCoord_of_comp_proj}
  (which computes \lean{fderiv} of $f \circ \mathrm{proj}_0$).
\end{corollary}

\begin{corollary}[Rectangular Green's theorem, \lean{green_stokes}]\label{cor:green}
  For smooth $Q, P : \RR^2 \to \RR$ and a rectangle $\Icc{a}{b}$, writing
  $[Q, P]$ for the coordinate 1-form with coefficients $(\omega_0, \omega_1) = (Q, P)$:
  \[
    \int_{\Icc{a}{b}} \left(\frac{\bdry Q}{\bdry x} - \frac{\bdry P}{\bdry y}\right)
      \dif x\,\dif y = \int_{\bdry\Icc{a}{b}} [Q, P]
  \]
  The coordinate convention is: a 1-form on $\RR^2$ in coordinates
  $(x_0, x_1) = (x, y)$ is $\omega = \omega_0\,\dif x_1 + \omega_1\,\dif x_0$
  where $\omega_i$ is the coefficient of the 1-form that \emph{omits} $\dif x_i$.
  Thus $\omega_0 = Q$ (paired with $\dif y$ since $\dif x_0 = \dif x$ is omitted)
  and $\omega_1 = P$ (paired with $\dif x$ since $\dif x_1 = \dif y$ is omitted),
  recovering the classical $\omega = P\,\dif x + Q\,\dif y$.
\end{corollary}

\begin{corollary}[Divergence consistency, \lean{divergence_stokes}]\label{cor:div}
  For a smooth vector field $F : \RR^{n+1} \to \RR^{n+1}$ and a box $\Icc{a}{b}$:
  \[
    \int_{\Icc{a}{b}} \mathrm{div}(F)\,\dif V
      = \sum_{i=0}^{n} \left(
        \int_{\mathrm{face}_i} F_i(b_i, \hat{x}_i)\,\dif\hat{x}_i
        - \int_{\mathrm{face}_i} F_i(a_i, \hat{x}_i)\,\dif\hat{x}_i
      \right)
  \]
  The signs cancel: the form uses $\omega_i = (-1)^i F_i$ and \lean{signedCoeff}
  applies another $(-1)^i$, so the boundary integrand is
  $(-1)^{2i} F_i = F_i$ (since $(-1)^{2i} = 1$).
  This is a \emph{consistency check}: since our main theorem is proved FROM the box
  divergence theorem, recovering the divergence theorem from Stokes shows the
  sign conventions round-trip correctly. It is not an independent derivation.
\end{corollary}

\begin{corollary}[Rectangular Green's four-edge decomposition, \lean{green_four_edges}]\label{cor:green-edges}
  For the 1-form $\omega = P\,\dif x + Q\,\dif y$ (standard convention
  with $\extd\omega = (\partial_x Q - \partial_y P)\,\dif x\wedge\dif y$),
  the double integral decomposes into four oriented face integrals:
  \[
    \iint_R \left(\frac{\partial Q}{\partial x} - \frac{\partial P}{\partial y}\right) dA
    = \int_{\text{right}} Q - \int_{\text{left}} Q - \int_{\text{top}} P + \int_{\text{bottom}} P
  \]
  where right ($x{=}b_0$), left ($x{=}a_0$), top ($y{=}b_1$), bottom ($y{=}a_1$).
  This gives the explicit face-integral decomposition of the boundary sum.
\end{corollary}

\begin{corollary}[Integration by parts, \lean{integration_by_parts}]\label{cor:ibp}
  For smooth $f, g : \RR \to \RR$ and $a \leq b$:
  \[
    \int_a^b f(x) g'(x)\,\dif x = f(b)g(b) - f(a)g(a) - \int_a^b f'(x)g(x)\,\dif x
  \]
  Derived from \lean{ftc_stokes} applied to $f \cdot g$ together with the product rule.
\end{corollary}

\subsection{1D Agreement Check}
\label{sec:validation}

A natural concern with the coordinate approach is whether the Stokes theorem
produces correct integral values. We validate this for the 1D case ($n{=}0$, the FTC)
by computing via two
independent paths and showing they agree:
\begin{enumerate}
  \item \textbf{Path 1 (Stokes):} \lean{stokes_smooth} $\to$ \lean{ftc_stokes}
    $\to$ $f(b) - f(a)$
  \item \textbf{Path 2 (interval integral):}\\
    \lean{intervalIntegral.integral_eq_sub_of_hasDerivAt}
    $\to$ \lean{ftc_intervalIntegral} $\to$ $f(b) - f(a)$
\end{enumerate}

The agreement theorem (\lean{ftc_paths_agree}) shows both paths produce the same result.
The Stokes path uses set integrals over $\mathrm{Fin}\,1 \to \RR$ of $\mathrm{fderiv}$,
while the interval-integral path uses mathlib's direct FTC
(\lean{intervalIntegral.integral_eq_sub_of_hasDerivAt}).
Their agreement provides a sanity check that our box-integral machinery computes
standard calculus integrals correctly in the $n{=}0$ case, though both paths ultimately
rely on mathlib's analysis library (the divergence theorem for path 1, direct FTC for path 2).
The validation is therefore a consistency check of our definitions and sign
conventions against mathlib's independent FTC path, not a check of mathlib's
analysis foundations themselves.

\section{Formal Sums and Linear Extension}
\label{sec:chains}

We define formal $\ZZ$-linear combinations of boxes and extend Stokes by linearity:

\begin{definition}[Cubical box and chain]
  A \emph{cubical box} is a structure $(a, b, h)$ where $a \leq b$ componentwise.
  A \emph{cubical chain} is a formal $\ZZ$-linear combination of boxes:
  \lean{CubicalChain n := CubicalBox n ->0 Z} (where \texttt{->0} denotes
  finitely-supported functions, i.e., \lean{Finsupp}).
\end{definition}

Integration extends linearly to chains:
\begin{lstlisting}
def integrateExterior (c : CubicalChain n) (omega : CoordNForm n) : Real :=
  c.sum fun B k => (k : Real) * boxIntegral (extDerivCoord omega) B.lo B.hi
\end{lstlisting}

\begin{theorem}[Stokes for chains, \lean{stokes_chain}]\label{thm:stokes-chain}
  For any smooth form $\omega$ and chain $c$:
  \[\mathrm{integrateExterior}(c, \omega) = \mathrm{integrateBdry}(c, \omega)\]
  The proof is by \lean{Finsupp.sum_congr}, applying \lean{stokes_smooth} to each box.
\end{theorem}

Additivity properties (\lean{integrateExterior_add}, \lean{integrateBdry_add}) are
proved via \lean{Finsupp.sum_add_index'}.

\subsection{Two-Box Additivity and Shared-Face Cancellation}
\label{sec:subdivision}

A natural question is whether the Stokes identity is compatible with decomposing a box
into two adjacent sub-boxes.
The subdivision module shows that for any two boxes (forming a 2-element chain), the
Stokes identity holds for the combined chain:

\begin{theorem}[Two-box Stokes additivity, \lean{subdivision_stokes_equiv}]
  For two boxes $B_1$, $B_2$ (with $a_i \le b_i$) and a smooth form $\omega$:
  \[
    \int_{B_1} d\omega + \int_{B_2} d\omega = \int_{\partial B_1} \omega + \int_{\partial B_2} \omega
  \]
  When $B_1$ and $B_2$ share a face (subdivision of a larger box), the shared face
  contributions cancel in the boundary sum, recovering the Stokes identity for the
  two-box formal chain.
\end{theorem}

This is proved as a direct consequence of the single-box Stokes theorem (by \lean{linarith}).
The chain-level version (\lean{stokes_two_boxes}) follows from \lean{stokes_chain}.

We additionally formalize the key \emph{shared-face cancellation}:

\begin{definition}[Adjacency, \lean{Adjacent}]
  Two boxes $B_1, B_2$ are \textbf{adjacent in direction $i$} if
  $B_1.\mathrm{hi}[i] = B_2.\mathrm{lo}[i]$ and the remaining coordinates agree.
\end{definition}

\begin{theorem}[Shared-face cancellation, \lean{shared_face_cancel}]
  If $B_1$ and $B_2$ are adjacent in direction $i$, then the integral over
  the high-face of $B_1$ at direction $i$ equals the integral over the
  low-face of $B_2$ at direction $i$:
  \[
    \int_{\mathrm{face}}\!f_i|_{x_i = B_1.\mathrm{hi}[i]} =
    \int_{\mathrm{face}}\!f_i|_{x_i = B_2.\mathrm{lo}[i]}
  \]
  Since these appear with opposite signs in the combined boundary sum, they cancel.
\end{theorem}

\begin{corollary}[\lean{adjacent_boundary_simplifies}]
  The net contribution of the shared face to the combined boundary is zero.
\end{corollary}

\subsection{Paired Codimension-2 Face Cancellation}
\label{sec:bdry-sq}

The fundamental property of the boundary operator in a chain complex is $\partial^2 = 0$:
applying the boundary twice gives zero. We prove the \emph{sign-arithmetic core} of this
identity: for each codimension-2 face in the double boundary, the two paths that reach it
produce opposite signed contributions. This is the essential algebraic content of $\partial^2 = 0$,
though we do not define $\partial$ as a chain-group homomorphism (our undecorated box type
does not distinguish upper/lower faces as separate chain elements).

\begin{theorem}[Sign cancellation, \lean{bdry_sq_sign_cancel}]\label{thm:bdry-sq}
  For any $n$, indices $i : \mathrm{Fin}(n+3)$, $j : \mathrm{Fin}(n+2)$, and
  orientation booleans $\varepsilon, \eta$:
  \[
    (-1)^{(i + \varepsilon)} \cdot (-1)^{(j + \eta)} +
    (-1)^{(\mathrm{succAbove}(i,j) + \eta)} \cdot (-1)^{(\mathrm{predAbove}(j,i) + \varepsilon)} = 0
  \]
  The proof uses the parity lemma (\lean{succAbove_predAbove_neg_one_pow}):
  \[
    (\mathrm{succAbove}(i,j) + \mathrm{predAbove}(j,i)) \not\equiv (i + j) \pmod{2}
  \]
  which is proved by unfolding \lean{succAbove}/\lean{predAbove} to natural number
  arithmetic and using \lean{split_ifs <;> omega}, the same technique used by
  mathlib's \lean{succAbove_succAbove_succAbove_predAbove}.
\end{theorem}

\begin{theorem}[Geometric + algebraic, \lean{bdry_sq_geometric}]\label{thm:bdry-geom}
  For a box $B$ in $\RR^{n+3}$, the two paths to each codimension-2 face produce:
  \begin{enumerate}
    \item The same face box (by the face-of-face identity), and
    \item Opposite signs (by sign cancellation).
  \end{enumerate}
  This establishes the sign-arithmetic component of $\partial^2 = 0$ in cubical
  homology. A full $\partial^2 = 0$ would additionally require defining $\partial$
  as a map between chain groups, which requires distinguishing upper/lower faces
  (not possible with our undecorated box type).
\end{theorem}

\section{Algebraic Properties}
\label{sec:properties}

The exterior derivative \lean{extDerivCoord} is a linear operator:

\begin{proposition}[Linearity]\label{prop:linearity}
  For differentiable forms $\omega, \eta$ and scalar $c \in \RR$:
  \begin{align}
    \extd(\omega + \eta) &= \extd\omega + \extd\eta \\
    \extd(c \cdot \omega) &= c \cdot \extd\omega \\
    \extd(-\omega) &= -\extd\omega \\
    \extd(0) &= 0
  \end{align}
  Each is proved by the corresponding linearity property of \lean{fderiv}
  (\lean{fderiv_add}, \lean{fderiv_const_mul}, etc.).
\end{proposition}

Combined with Stokes, these give:
\begin{lstlisting}
theorem stokes_smooth_add (a b) (hle : a <= b) (omega eta) (ho he) :
    boxIntegral (extDerivCoord (CoordNForm.add omega eta)) a b =
    bdryIntegral (CoordNForm.add omega eta) a b
\end{lstlisting}

\begin{proposition}[Product Rule / Leibniz]\label{prop:leibniz}
  For a smooth scalar function $f$ and smooth coordinate form $\omega$:
  \[
    \extd(f \cdot \omega)(x) = f(x) \cdot \extd\omega(x) +
    \sum_i (-1)^i \frac{\partial f}{\partial x_i}(x) \cdot \omega_i(x)
  \]
  This is the coordinate expression of the graded Leibniz rule
  $d(f \wedge \omega) = df \wedge \omega + f \cdot d\omega$; the $(-1)^i$ factor
  arises from wedging $df$ (a 1-form) with the basis $n$-form
  $dx_0\wedge\cdots\wedge\widehat{dx_i}\wedge\cdots\wedge dx_n$.
  The proof uses \lean{HasFDerivAt.mul} from \mathlib{}.
\end{proposition}

\begin{proposition}[Concrete Computations]\label{prop:examples}
  We verify the formalization on concrete forms:
  \begin{itemize}
    \item $d(\text{const}) = 0$: constant forms are closed.
    \item $d(x_0) = 1$ in 1D: the exterior derivative of the coordinate
      function recovers the constant 1 (the derivative underlying FTC).
  \end{itemize}
\end{proposition}

\section{Coefficient Precomposition by Linear Maps}
\label{sec:naturality}

For coordinate forms on boxes, we additionally provide precomposition of
coefficients by continuous linear maps and prove functoriality. We emphasize
that this is \emph{coefficient-wise} precomposition
$(A \cdot \omega)_i(x) := \omega_i(Ax)$, which is \textbf{not} the true
exterior-algebraic pullback used in \Cref{sec:singular-stokes}: the true pullback
acts via the full Fr\'echet derivative and mixes coefficients, while
coefficient precomposition merely pre-composes each coefficient function
individually.

The development includes the expected lemmas:
smoothness preservation under continuous linear $A$ (\lean{pullback_smooth},
proved via \lean{ContDiff.comp});
contravariant functoriality
$(A \cdot (B \cdot \omega)) = (B \circ A) \cdot \omega$ (\lean{pullback_comp});
identity $\mathrm{id}\cdot\omega = \omega$ (\lean{pullback_id});
linearity in $\omega$ (\lean{pullback_add}, \lean{pullback_smul});
and a Stokes specialization for precomposed forms (\lean{stokes_pullback}),
which follows immediately from \lean{stokes_smooth} applied to $A\cdot\omega$.
We do \emph{not} claim that this operation commutes with the exterior derivative
in general; it satisfies Stokes only because Stokes holds for every smooth
coordinate form, regardless of its origin.

\begin{remark}
The Lean declaration names (\lean{pullback_comp}, \lean{pullback_id}, etc.)
use the term ``pullback'' for historical reasons, but these are coefficient-wise
precomposition lemmas, \emph{not} differential-form pullbacks. The true
differential-form pullback via \lean{fderiv} appears only in the singular cube
module (\Cref{sec:singular-stokes}).
\end{remark}

\section{The \texorpdfstring{$d^2{=}0$}{d²=0} Identity from \mathlib{}}
\label{sec:dd-zero}

A fundamental property of the exterior derivative is nilpotency: $d(d\omega) = 0$.
We derive this at the coordinate level through our bridge to \mathlib{}:

\begin{theorem}[$d^2{=}0$, \lean{dd_zero_abstract}]
  For any smooth $n$-form $\omega$ on $\RR^{n+1}$ (valued in
  $n$-fold continuous alternating maps $\RR^{n+1} \to \RR$),
  the iterated abstract exterior derivative vanishes:
  $\mathrm{extDeriv}(\mathrm{extDeriv}\;\omega) = 0$.
  The proof is a direct invocation of \lean{extDeriv_extDeriv} from \mathlib{}.
\end{theorem}

\begin{corollary}[Vanishing of double exterior derivative integral, \lean{stokes_dd_zero}]
  For any smooth $\omega$ and vector tuple $v$:
  $\int_{\Icc{a}{b}} (\extd(\extd\omega))(x)(v)\,\dif x = 0$.
  This follows immediately from $\extd^2 = 0$: the integrand is identically zero.
\end{corollary}

\subsection{Boundary Operator: Face-of-Face Identities}

We formalize the sign-arithmetic and geometric core underlying $\partial^2 = 0$
at the box level. For a box $B \subset \RR^{n+3}$,
the double boundary $\partial^2 B$ consists of codimension-2 faces, each
appearing via two paths:
\begin{itemize}
  \item Path 1: face $j$ of (face $i$ of $B$)
  \item Path 2: face $(\mathrm{predAbove}\;j\;i)$ of
    (face $(\mathrm{succAbove}\;i\;j)$ of $B$)
\end{itemize}

\begin{theorem}[Face-of-face, \lean{double_boundary_same_face}]
  The two paths give the same geometric box:
  \[
    \mathrm{face}_j(\mathrm{face}_i(B))
      = \mathrm{face}_{\mathrm{predAbove}(j,i)}\bigl(
          \mathrm{face}_{\mathrm{succAbove}(i,j)}(B)\bigr).
  \]
\end{theorem}

\begin{theorem}[Sign cancellation, \lean{double_boundary_cancels}]
  The signed coefficients for the two paths sum to zero:
  \[
    (-1)^{i+\varepsilon}(-1)^{j+\eta} + (-1)^{\mathrm{succAbove}(i,j)+\eta}
    (-1)^{\mathrm{predAbove}(j,i)+\varepsilon} = 0.
  \]
\end{theorem}

Together these establish the face-of-face composition identity and
sign-arithmetic core of $\partial^2 = 0$: each
codimension-2 face appears with opposite signs in the double boundary.
However, this is \emph{not} a full $\partial^2 = 0$ at the chain level, since we do not
define $\partial$ as a chain-group homomorphism (our undecorated box type does not
distinguish upper/lower faces as separate chain elements).

\subsection{Chain-Level \texorpdfstring{$\partial^2 = 0$}{∂²=0} for Singular Cubes}
\label{sec:bdry-sq-singular}

In the smooth singular cube module, we \emph{do} prove the full chain complex identity
$\partial^2 = 0$ for singular chains (cf.~the textbook treatment in~\cite{bott-tu}).
Here the boundary operator $\partial$ is a genuine
$\ZZ$-linear map $\partial : C_{n+1}(\RR^m) \to C_n(\RR^m)$ between singular chain groups.

\begin{theorem}[$\partial^2 = 0$, \lean{bdry_bdry_chain_zero}]\label{thm:bdry2-chain}
  For any smooth singular $(n{+}2)$-cube $\sigma$:
  \[
    \partial(\partial\sigma) = 0 \quad\text{as a formal singular $n$-chain.}
  \]
\end{theorem}

\paragraph{Proof strategy.}
The proof uses a sign-reversing fixed-point-free involution on
the index set $\mathrm{Fin}(n{+}2) \times \mathrm{Fin}(n{+}1)$ of the double sum,
combined with \lean{Finset.sum_ninvolution} from \mathlib{}.

The involution maps:
\[
  g(i, j) =
  \begin{cases}
    (j{+}1,\; i) & \text{if } i \leq j \\
    (j,\; i{-}1) & \text{if } i > j
  \end{cases}
\]

\noindent
For each pair $(i,j)$ with $i \leq j$, the cancellation argument has two components:
\begin{enumerate}
  \item \textbf{Same cube}: The face-of-face identity (\lean{singularFace_comp})
    shows that $\mathrm{face}_j(\delta, \mathrm{face}_i(\varepsilon, \sigma))
    = \mathrm{face}_i(\varepsilon, \mathrm{face}_{j+1}(\delta, \sigma))$,
    so paired terms produce the same singular cube.
  \item \textbf{Opposite sign}: The coefficients $(-1)^i(-1)^j$ and
    $(-1)^{j+1}(-1)^i$ satisfy $(-1)^{i+j} + (-1)^{i+j+1} = 0$.
\end{enumerate}

After applying the face-of-face identity, the proof uses
\texttt{ext $\tau$; simp [Finsupp.single\_apply]; split\_ifs} \texttt{<;> ring}
to close all 256 goals of the split case analysis automatically
(goal count observed with Lean~4.29.1; may vary across versions).

\begin{theorem}[$\partial^2 = 0$ for chains, \lean{bdry_bdry_chain_zero_general}]
  For any singular chain $c \in C_{n+2}(\RR^m)$:
  \[
    \partial(\partial c) = 0
  \]
  The proof extends from individual cubes by induction on the chain support,
  using the $\ZZ$-linearity of $\partial$.
\end{theorem}

\begin{corollary}[Integral vanishing, \lean{bdry_bdry_integral_zero}]
  For any smooth singular $(n{+}2)$-cube $\sigma$ and $n$-form $\omega$:
  \[
    \int_{\partial(\partial\sigma)} \omega = 0
  \]
\end{corollary}

This establishes that $(\{C_n(\RR^m)\}, \partial)$ forms a genuine chain complex.
Combined with \lean{stokes_chain}, this yields the identity one would use to
define integration as a cochain map from the de~Rham complex to singular cubical
cochains; however, we do not formally package these complexes or prove the
commuting-square diagram as a Lean structure.

\section{Architecture and Infrastructure}
\label{sec:architecture}

\subsection{Module Structure}

The formalization consists of 44 Lean~4 source files in the \texttt{LeanStokes/}
directory totaling 4028 source lines, plus a root import file
\texttt{LeanStokes.lean}. Selected modules and their roles:

\begin{center}
\begin{small}
\begin{tabular}{lll}
\toprule
\textbf{Module} & \textbf{Content} & \textbf{Key Theorems} \\
\midrule
\multicolumn{3}{l}{\textit{Singular Cube Stokes}} \\
\texttt{.../Defs} & Singular cubes, faces & \lean{SmoothSingularCube}, \lean{singularFace} \\
\texttt{.../Pullback} & True pullback via fderiv & \lean{pullbackForm}, \lean{integrateForm} \\
\texttt{.../Smoothness} & Pullback smoothness & \lean{toCoordNForm_pullback_isSmooth} \\
\texttt{.../FaceMatching} & Face matching identity & \lean{face_matching} \\
\texttt{.../Theorem} & Main singular Stokes & \lean{singularStokes}, \lean{pullback_extDeriv} \\
\texttt{.../Chain} & Singular chains & \lean{stokes_singular_boundary} \\
\texttt{.../BdryBdry} & $\partial^2{=}0$ chain-level & \lean{bdry_bdry_chain_zero} \\
\midrule
\multicolumn{3}{l}{\textit{Box Stokes Infrastructure}} \\
\texttt{Defs} & Core definitions & \lean{CoordNForm}, \lean{extDerivCoord} \\
\texttt{Theorem} & Main box identity & \lean{stokes_on_box} \\
\texttt{Smooth} & Smooth-form version & \lean{stokes_smooth} \\
\texttt{Classical} & Dim.\ specializations & \lean{green_on_rectangle} \\
\texttt{Bridge} & extDeriv connection & \lean{extDeriv_topCoeff_eq_extDerivCoord} \\
\texttt{Unified} & Combined statement & \lean{stokes_extDeriv_smooth} \\
\texttt{FTC} & Fund.\ Thm.\ of Calculus & \lean{ftc_stokes} \\
\texttt{Green} & Green's theorem & \lean{green_stokes} \\
\texttt{Divergence} & Divergence consistency & \lean{divergence_stokes} \\
\texttt{IBP} & Integration by parts & \lean{integration_by_parts} \\
\texttt{Chains} & Formal sums + linearity & \lean{stokes_chain} \\
\texttt{Subdivision} & Two-box decomposition & \lean{subdivision_stokes_equiv} \\
\texttt{Naturality} & Coeff.\ precomp.\ functoriality & \lean{pullback_comp} \\
\texttt{DDZero} & $d^2{=}0$ from \mathlib{} & \lean{dd_zero_abstract} \\
\texttt{BdryOp} & $\partial^2{=}0$ sign arith. & \lean{bdry_sq_sign_cancel} \\
\texttt{Gauss3D} & 3D divergence theorem & \lean{gauss_3d} \\
\texttt{Demo} & Usage examples & 4D/5D/10D Stokes \\
\texttt{Check} & Axiom verification & 79 \lean{\#print axioms} calls \\
\bottomrule
\end{tabular}
\end{small}
\end{center}

All modules are sorry-free.

\subsection{Design Decisions}

\paragraph{Dimension indexing.}
We use \lean{Fin (n + 1)} for the ambient dimension, avoiding the subtraction
issue where \texttt{(d - 1) + 1} $\neq$ \texttt{d} in natural number arithmetic.

\paragraph{Signed coefficients.}
The key insight enabling the 3-line proof is defining
$f_i(x) = (-1)^i \omega_i(x)$, so that:
\begin{itemize}
  \item $\mathrm{div}(f) = \sum_i \partial f_i / \partial x_i$, which equals the
    single coefficient of the top-form $\extd\omega$
  \item The divergence theorem's boundary terms match our boundary integral
\end{itemize}

\paragraph{Separation of concerns.}
The general theorem (\lean{stokes_on_box}) takes explicit hypotheses about
continuity, differentiability, and integrability. The smooth version
(\lean{stokes_smooth}) derives these automatically from \lean{ContDiff Real Top}.

\subsection{Axiom Audit}

Every key theorem was verified via \lean{\#print axioms} (79 declarations
checked, covering all main results and their transitive dependencies).
All produce the same axiom set:
\begin{itemize}
  \item \texttt{propext} (propositional extensionality)
  \item \texttt{Classical.choice} (classical logic)
  \item \texttt{Quot.sound} (quotient soundness)
\end{itemize}
These are the standard axioms of Lean~4's foundation, present in all
nontrivial \mathlib{} proofs. No \texttt{sorry} or additional axioms are used.

\section{Formalization Insights and Challenges}
\label{sec:insights}

This section describes what the formalization process revealed about the mathematics
of Stokes' theorem and cubical chains. We present both ``what was proved'' and
``what was learned,'' following established practice in formalization papers~\cite{best-flt,van-doorn-itp}.

\subsection{The Bridge Pattern: Connecting Concrete and Abstract}

The most reusable methodological contribution of this work is the \emph{bridge theorem}
pattern. \mathlib{}'s \lean{extDeriv} operates on \lean{ContinuousAlternatingMap}-valued
functions in full generality, while our box-level infrastructure uses coordinate-indexed
functions $\omega_i : \RR^{n+1} \to \RR$. The bridge theorem
(\lean{stokes_extDeriv_smooth}, \Cref{sec:bridge}) proves that evaluating the abstract
\lean{extDeriv} on standard basis vectors recovers our coordinate formula.

This pattern---where a concrete computational representation is validated against
an abstract library definition---may be reusable wherever formalizers face the choice
between ``build from scratch in a concrete representation'' and ``use the library's
abstract API directly''~\cite{van-doorn-measure}. The bridge makes both approaches compatible.

\subsection{Pitfall: Nat Subtraction in Sign Exponents}

A recurring obstacle in the chain-level $\partial^2 = 0$ proof was Lean~4's treatment
of natural-number subtraction in exponents. The ring tactic cannot simplify
$(-1)^{n-1}$ when $n$ is a natural-number variable, because $n - 1$ is not a
polynomial operation (it truncates to 0 when $n = 0$).

Our solution was the helper lemma \lean{neg_one_pow_pred}: for $n > 0$,
$(-1 : \ZZ)^{n-1} = -(-1)^n$. This allows \lean{simp only [neg_one_pow_pred ...]}
to preprocess sign expressions before \lean{ring} closes the goal.
This illustrates a general principle: \emph{ring-normalizable expressions must avoid
Nat subtraction}, and dedicated lemmas are needed to bridge the gap.

\subsection{Pitfall: Fin Projection Opacity}

Lean~4's kernel does not always reduce $(\langle e, h\rangle : \mathrm{Fin}\;n).val$
to $e$ during unification. This caused repeated failures when proving equalities
involving composed face maps (e.g., $\mathrm{faceInclusion}\;i\;\varepsilon \circ
\mathrm{faceInclusion}\;j\;\delta$). The solution is explicit use of
\lean{Fin.val_mk} or \lean{dsimp} to force projection reduction before
\lean{omega} can close the arithmetic goal.

\subsection{Insight: Why Global Smoothness Simplifies Everything}

We require singular cubes to be globally $C^\infty$ (not just smooth on $[0,1]^{n+1}$).
This is mathematically stronger than necessary but has a cascade of simplification effects:
\begin{enumerate}
  \item \lean{ContDiff} composes freely: $\sigma \circ \iota_{i,\varepsilon}$ is
    automatically $C^\infty$ without boundary regularity arguments.
  \item \lean{fderiv} is well-defined everywhere, avoiding the need for
    \lean{HasFDerivAt} hypotheses localized to the cube interior.
  \item The pullback $\sigma^*\omega$ is smooth on all of $\RR^{n+1}$, so
    box Stokes applies directly without restriction arguments.
\end{enumerate}
A formalization that required smoothness only on the cube would need
\lean{ContDiffOn} (which composes less freely) and careful bookkeeping of
open neighborhoods---substantial additional infrastructure with no
mathematical payoff for the theorem statement.

\subsection{Insight: extDeriv\_pullback as the Key Enabler}

Without \mathlib{}'s \lean{extDeriv_pullback}, the singular Stokes theorem
would have been substantially harder. That single theorem provides
$d(\sigma^*\omega) = \sigma^*(d\omega)$ pointwise and reduces singular Stokes
to box Stokes in three steps: (1)~apply naturality, (2)~apply box Stokes to
the pullback, (3)~match faces. Without it we would have had to reprove
naturality of the exterior derivative under pullback---a substantial
undertaking involving the Leibniz rule for wedge products and the
composition of multilinear maps.

\subsection{The Involution Method for \texorpdfstring{$\partial^2 = 0$}{∂²=0}}

The chain-level $\partial^2 = 0$ proof uses \lean{Finset.sum_ninvolution}: to show
$\sum_{p \in S} f(p) = 0$, it suffices to find a fixed-point-free involution
$g : S \to S$ such that $f(p) + f(g(p)) = 0$ for all $p$. This is a
standard combinatorial technique, but its application to cubical boundary algebra
required careful construction: the involution $g(i,j) = (j{+}1, i)$ when $i \leq j$
and $(j, i{-}1)$ otherwise produces 256 case-split goals in the cancellation
proof, all of which close with \lean{ring} after sign preprocessing.

The mathematical lesson is that $\partial^2 = 0$ for cubical chains is fundamentally
a combinatorial identity (sign-reversing involution on index pairs) rather than
an algebraic identity (in the sense of ring operations). The formal proof makes
this distinction precise in a way informal proofs typically obscure.

\section{Discussion}
\label{sec:discussion}

\subsection{Scope and Limitations}

Our formalization proves Stokes' theorem at two levels: (1)~for \emph{smooth singular
cubes}---globally $C^\infty$ maps $\sigma : \RR^{n+1} \to \RR^m$ integrated over
$[0,1]^{n+1}$ with true pullback via \lean{fderiv}---and (2)~for axis-aligned
boxes and their formal $\ZZ$-linear combinations. We state the concrete limitations
explicitly:

\begin{itemize}
  \item \textbf{No manifolds or partitions of unity}: The integration domain is always
    $[0,1]^{n+1}$ (or $\Icc{a}{b}$ for boxes). No charts, atlases, smooth structures,
    or partition-of-unity arguments are used. Singular cubes must be single smooth
    parametrizations.
  \item \textbf{Global smoothness}: Singular cubes are required to be $C^\infty$ on
    all of $\RR^{n+1}$, not just on $[0,1]^{n+1}$. This is stronger than necessary
    mathematically, but simplifies the formalization.
  \item \textbf{Box-level form degree}: The box-level infrastructure handles only
    coordinate $n$-forms on $\RR^{n+1}$ (one degree below top). The pullback
    construction is defined for arbitrary form degree~$k$; the singular Stokes
    theorem applies to $n$-forms integrated over smooth singular $(n{+}1)$-cubes.
  \item \textbf{Box Stokes reduces to divergence theorem}: The core box theorem reduces
    to \mathlib{}'s divergence theorem.
    The mathematical content is in the sign-convention construction, the bridge to
    \lean{extDeriv}, and the singular-cube extension via \lean{extDeriv_pullback}.
  \item \textbf{No image-domain semantics}: $\int_\sigma \omega$ is defined as an integral
    over $[0,1]^d$ of the pullback. We do not prove that $\sigma([0,1]^d)$ is a
    submanifold or oriented geometric chain. The singular-cube module does achieve
    a genuine chain-complex structure: $\partial^2 = 0$ is proved as a Finsupp identity
    (\lean{bdry_bdry_chain_zero}).
\end{itemize}

Despite these limitations, the formalization provides:
\begin{enumerate}
  \item \textbf{True pullback} of differential forms via the Fr\'echet derivative---the
    standard differential-geometric operation.
  \item \textbf{Smooth singular parametrizations}: Singular cubes provide smooth parametrizations
    of curves, surfaces, and higher-dimensional objects in $\RR^m$, not just axis-aligned boxes.
    We do not prove image-domain theory (e.g., that the image is a submanifold).
  \item Genuinely $n$-dimensional results for arbitrary~$n$ (instantiated in demos up to $n{=}9$, i.e., 10D).
  \item A \textbf{bridge theorem} validated against \mathlib{}'s abstract \lean{extDeriv}.
  \item \textbf{Two-box additivity} with shared-face cancellation showing composability.
  \item \textbf{1D agreement check} confirming correct integral values for the FTC case.
  \item \textbf{Reusable infrastructure}: the singular-cube framework, face matching,
    coordinate forms, and the bridge pattern may serve as building blocks for
    further development.
\end{enumerate}

\subsection{What Remains for Full Manifold Stokes}

A formalization of Stokes on smooth manifolds with boundary~\cite{lee-manifolds} would additionally require:
\begin{enumerate}
  \item \textbf{Integration of forms on manifolds}: defining $\int_M \omega$ via
    partition of unity and chart pullbacks.
  \item \textbf{Boundary orientation}: the induced orientation on $\bdry M$ from
    the outward normal.
  \item \textbf{Change of variables for forms}: proving that the integral is
    chart-independent.
  \item \textbf{Stokes on half-spaces}: the local compactly-supported version.
  \item \textbf{Singular chains on manifolds}: extending our singular cube
    infrastructure to manifold-valued maps.
\end{enumerate}
Each of these is a substantial formalization project in its own right. Our singular
cube framework provides a natural starting point: once manifold charts and partition
of unity are available, singular cubes may serve as the local parametrizations
over which Stokes is applied.

\subsection{Comparison with HOL Light}
\label{sec:harrison}

Harrison's HOL Light formalization~\cite{harrison-stokes} is the closest prior work.
A key finding of our investigation is that Harrison's development proves Stokes for
\emph{convex and polyhedral domains} extended to polyhedral chains---\emph{not} for
arbitrary smooth manifolds with boundary, as some secondary sources suggest.

With the addition of singular cubes, the mathematical scope of our development is
\emph{complementary to Harrison's}:
Harrison's singular simplices/polytopes are parametrized by smooth maps from
convex domains; our singular cubes are parametrized by smooth maps from the unit cube.
Harrison has strictly broader geometric domain coverage (convex polytopes with simplicial
decomposition); our contribution is a Lean~4/\mathlib{} formalization with direct bridge
to \lean{extDeriv} and reuse of \lean{extDeriv_pullback}.
Both use the true differential-form pullback.

\begin{table}[h]
\centering
\begin{tabular}{lll}
\toprule
\textbf{Aspect} & \textbf{Harrison (HOL Light)} & \textbf{This work (Lean~4)} \\
\midrule
Domain & Convex polytopes & Singular cubes + boxes \\
Parametrization & Smooth maps on polytopes & $C^\infty$ maps $\sigma : \RR^{n+1} \to \RR^m$ \\
Pullback & True (smooth maps) & True (via \lean{fderiv}) \\
Chains & Polyhedral chains & $\ZZ$-singular + cubical chains \\
Integration & Gauge (HK) & Bochner (set integral) \\
Ext.\ derivative & Custom definition & Via divergence + \lean{extDeriv_pullback} \\
$d^2{=}0$ & Algebraic & Via \lean{extDeriv_extDeriv} + chain involution \\
Bridge & None & To \lean{extDeriv} \\
Type theory & Simple (HOL) & Dependent (Lean~4) \\
Library reuse & Self-contained & Reduces to \mathlib{} \\
\bottomrule
\end{tabular}
\caption{Structural comparison between the two formalizations.}
\end{table}

The distinguishing feature of our approach is the bridge to \mathlib{}'s abstract exterior
algebra and the use of \lean{extDeriv_pullback} for the singular theorem. We are not
aware of a directly comparable bridge to an independently developed abstract
exterior-derivative API in HOL~Light, since HOL~Light does not maintain an independent
abstract exterior algebra alongside its integration theory.
This means our infrastructure is directly connected to the abstract theory and may
benefit from future growth in \mathlib{}'s manifold integration API.

Harrison's advantage is strictly broader domain generality: convex polytopes include
simplices and other non-cubical shapes. Our singular cubes always have the unit cube as the
parameter domain; the smooth map $\sigma$ can produce many parametrized images
(possibly degenerate or self-overlapping), but the theorem remains a statement about
integration over the parameter cube, not over the image as a domain.

\subsection{Reproducibility}

The complete formalization is available at:
\begin{center}
  \url{https://github.com/d0d1/lean-stokes-theorem}
\end{center}
A Software Heritage archive (SWHID) will be deposited prior to publication.
It builds with Lean~4.29.1 and \mathlib{} (pinned to \texttt{rev = "v4.29.1"} in
\texttt{lakefile.toml}; both version pins are in the repository). To verify:
\begin{verbatim}
  lake exe cache get && lake build
\end{verbatim}
Representative \lean{\#print axioms} output (all checked theorems produce
the same three axioms):
\begin{verbatim}
'SingularCubeStokes.singularStokes' depends on axioms:
  [propext, Classical.choice, Quot.sound]
'CubeStokes.stokes_on_box' depends on axioms:
  [propext, Classical.choice, Quot.sound]
'SingularCubeStokes.stokes_singular_boundary' depends on axioms:
  [propext, Classical.choice, Quot.sound]
\end{verbatim}

\section{Conclusion}

We have presented a sorry-free Lean~4 formalization of Stokes' theorem for smooth
singular cubes and axis-aligned boxes in arbitrary dimension, comprising 44 modules
with 205 named declarations across approximately 4000 source lines. All results
depend only on the standard axioms (\texttt{propext}, \texttt{Classical.choice},
\texttt{Quot.sound}), confirmed by 79 \lean{\#print axioms} checks covering
both the box and singular cubical layers.

The headline theorem---Stokes for smooth singular cubes with \emph{true}
differential-form pullback---reduces to the box theorem in three steps:
naturality of the pullback under the exterior derivative
($\extd(\sigma^*\omega) = \sigma^*(\extd\omega)$, from \mathlib{}'s
\lean{extDeriv_pullback}), box Stokes applied to $\sigma^*\omega$, and a
chain-rule face-matching identity. The same machinery extends by $\ZZ$-linearity
to singular chains, on which we prove $\partial(\partial\sigma)=0$ via a
sign-reversing involution on the double-boundary index set.

\paragraph{What we hope is reusable.}
Three ingredients seem most likely to be useful for downstream formalizations:
(i)~the \emph{bridge theorem} pattern, which validates a concrete coordinate
formula against an abstract library definition; (ii)~the \emph{face matching
identity} relating box-boundary integrals of pullback forms to face integrals
of singular cubes; and (iii)~the \emph{involution-based} proof of $\partial^2 = 0$
on singular cubical chains.

\paragraph{What remains.}
A full Lean~4 formalization of Stokes for smooth manifolds with boundary still
requires integration of forms via partition of unity and chart pullbacks,
boundary orientation, and image-domain singular chains. Our singular-cube
framework provides a natural local building block for such an effort once
\mathlib{}'s manifold integration API matures.

The artifact is reproducible via \texttt{lake exe cache get \&\& lake build}
and is available at \url{https://github.com/d0d1/lean-stokes-theorem} under
GPL-3.0-only, with versions pinned via \texttt{lean-toolchain} and
\texttt{lakefile.toml}.

\paragraph{Acknowledgments.}
We thank the \mathlib{} community for maintaining the mathematical library on which
this work depends, and the Lean~4 developers for the proof assistant infrastructure.

\end{document}